\RequirePackage{amsmath}
\documentclass[copyright,creativecommons]{eptcs}
\providecommand{\event}{AFL 2026} % Name of the event you are submitting to

\usepackage{iftex}

\ifpdf
  \usepackage{underscore}         % Only needed if you use pdflatex.
  \usepackage[T1]{fontenc}        % Recommended with pdflatex
\else
  \usepackage{breakurl}           % Not needed if you use pdflatex only.
\fi

\title{An Analysis of Decision Problems for Relational Pattern Languages under Various Constraints}
\author{Klaus Jansen\quad
Dirk Nowotka\quad
Lis Pirotton\quad
Corinna Wambsganz \quad
Max Wiedenhöft
\institute{Department of Computer Science, Kiel University, Kiel, Germany}
\email{\{kj,dn,lpi,cwa,maw\}@informatik.uni-kiel.de}
}
\def\titlerunning{Decision Problems for Relational Pattern Languages}
\def\authorrunning{K. Jansen et al.}

\usepackage[]{lineno}

\usepackage{amsfonts, amssymb, dsfont, amsthm}
\usepackage{graphicx}
\usepackage[utf8]{inputenc}
\usepackage[T1]{fontenc}
\usepackage{color}
\usepackage{float}
\usepackage{wasysym}
\usepackage{blkarray,booktabs,bigstrut}
\usepackage{verbatim}
\usepackage{tabularx}
\usepackage{upgreek}
\usepackage[nocompress]{cite}
\usepackage{hyperref}
\usepackage{cleveref}
\usepackage{enumitem}
\usepackage{marvosym}
\usepackage{booktabs}
\usepackage{tabularx}
\usepackage{array}
\usepackage[table]{xcolor}

\DeclareMathOperator{\SubSeq}{SubSeq}

\DeclareMathOperator{\var}{var}

\def\N{\mathbb{N}}
\def\ta{\mathtt{a}}
\def\tb{\mathtt{b}}
\def\tc{\mathtt{c}}

\def\relpat{\mathtt{RelPat}_{\Sigma,R}}

\newcommand{\al}{\operatorname{alph}}

\newcolumntype{C}{>{\hspace{3pt}}c<{\hspace{3pt}}}

\newtheorem{theorem}{Theorem}
\newtheorem{corollary}[theorem]{Corollary}
\newtheorem{lemma}[theorem]{Lemma}
\newtheorem{proposition}[theorem]{Proposition}
\newtheorem*{remark}{Remark}

\begin{document}
%\linenumbers

\maketitle

\begin{abstract}
Patterns are words with terminals and variables. The language of a pattern is the set of words obtained by uniformly substituting all variables with words that contain only terminals. In their original definition, patterns only allow for multiple distinct occurrences of some variables to be related by the equality relation, represented by using the same variable multiple times. In an extended notion, called relational patterns and relational pattern languages, variables may be related by arbitrary other relations, achieved by using regular patterns and relating individual variables independently from the patterns structure separately. We extend the ongoing investigation of the main decision problems for patterns (namely, the equivalence problem, the inclusion problem, and the membership problem) to relational pattern languages under a wide range of relevant individual relations, providing a comprehensive foundation in all three research directions.
\vspace{0.2cm}

\end{abstract}

% === INTRO === %
\section{Introduction}
A \emph{pattern} (with variables) is a finite word that consists only of symbols from a finite set of \emph{(terminal) letters} $\Sigma = \{\ta_1, ..., \ta_\sigma\}$ and from an infinite but countable set of \emph{variables} $X = \{x_1, x_2, ...\}$ with $\Sigma\cap X = \emptyset$. By applying terminal preserving morphisms, called \emph{substitutions}, to a pattern, we can obtain words that consist only of terminal letters. The \emph{language} of a pattern is then just the set of all words that can be obtained through the application of arbitrary substitutions.
Originally, pattern languages as defined and introduced by Angluin~\cite{Angluin1980} only considered words that are obtained by what is now called \emph{non-erasing substitutions}, i.e., substitutions where each variable is substituted by a word of at least length one. Hence, these languages are now also called \emph{non-erasing (NE) pattern languages}. Shinohara~\cite{Shinohara1983} extended this notion by allowing for the empty word to be used to substitute a variable, introducing so-called \emph{extended/erasing (E) pattern languages}. 

For example, consider the pattern $\alpha = x_1 \ta\ta x_1 \tb x_2$. Using the non-erasing substitution $h$ that sets $h(x_1) = \tb\tb$ and $h(x_2)= \ta$, we obtain the word $h(\alpha) = \tb\tb\ta\ta\tb\tb\tb\ta$. If we also considered the E-pattern language of $\alpha$, we could set all variables to the empty word and obtain the word $\ta\ta\tb$ which we could not obtain in the NE pattern language of $\alpha$.

As patterns form a natural and compact device for the definition of formal languages and due to their practical and simple definition, they occur in numerous areas in computer science and discrete mathematics. Among others, these include for example the areas of unavoidable patterns in words~\cite{Jiang1994,lothaire1997}, word equations~\cite{lothaire1997}, algorithmic learning theory~\cite{Angluin1980,FERNAU201844,Shinohara1995}, database theory~\cite{FreydenbergerP21,SchmidSchweikardtPODS2022}, or the theory of extended regular expressions with back references\cite{FREYDENBERGER20191}. Furthermore, many practical areas such as machine learning, database systems, or bio-informatics deal with the question of identifying or describing patterns in (sets of) strings~\cite{Mousawi2024}, for example in the analysis of protein data~\cite{Arikawa1993}, in pattern matching~\cite{10.1007/978-3-642-03784-9_29}, or the design of algorithms for program synthesis~\cite{10.1145/4472.4476}.

Three main decision problems emerge when it comes to the investigation of patterns and pattern languages. These are the general \emph{equivalence problem}, the \emph{inclusion problem}, and the \emph{membership problem} and its variations~\cite{gawrychowski_et_al:LIPIcs.MFCS.2021.48,Manea2022, Fleischmann2023}, all of which are considered in the non-erasing (NE) as well as erasing (E) cases. The equivalence problem determines whether the languages of two given patterns are equal. In Angluin's original work~\cite{Angluin1980}, it has been shown to be trivially decidable for all NE pattern languages. For erasing pattern languages, however, the general decidability of this problem has been the \emph{most significant open problem} in the field of decision problems for patterns~\cite{Jiang1995,Reidenbach2004-1,Reidenbach2004-2,Ohlebusch1996,Reidenbach2007}. The inclusion problem checks whether the language of a pattern is included as a subset in another pattern's language. It has been shown to be undecidable for unbounded alphabets by Jiang et al.~\cite{Jiang1995}. Their result has been extended by Freydenberger and Reidenbach in~\cite{Freydenberger2010} as well as by Bremer and Freydenberger in~\cite{BREMER201215}, where it has been shown to be undecidable for all alphabets $\Sigma$ of fixed size $|\Sigma|\geq 2$, for erasing and non-erasing pattern languages. The membership problem determines whether a word belongs to a pattern's language. It has been shown to be generally NP-complete in the non-erasing as well as erasing cases~\cite{Angluin1980,Jiang1994} and has been investigated extensively in various communities~\cite{AMIR2007514, Ordyniak2015, FERNAU2015287, Fernau2016, IBARRA1995179, Manea2019, Manea2020}.
In the special case of terminal-free patterns (containing only variables), the membership problem remains hard in both erasing and non-erasing cases, as does the non-erasing equivalence problem. However, inclusion for terminal-free E-pattern languages is NP-complete, as it reduces to checking for a morphism between patterns~\cite{Jiang1995,DBLP:journals/ipl/EhrenfeuchtR79a,File1988}, resulting also in decidability of equivalence. In contrast, Saarela~\cite{saarela:LIPIcs.ICALP.2020.140} showed that inclusion for terminal-free NE pattern languages is undecidable over a binary alphabet.

Various extensions to the notion of pattern languages have been introduced over time either to get closer to an answer for the remaining open problems or to obtain additional expressibility that is usable in some practical context. Some examples include the bounded scope coincidence degree, patterns with bounded treewidth, $k$-local patterns, or strongly-nested patterns (see \cite{Day2018} and references therein). Koshiba~\cite{Koshiba1995} introduced so-called \emph{typed patterns} where the substitutions of variables can be restricted by arbitrary recursive languages, so-called \emph{types}. This has been recently extended by Geilke and Zilles in~\cite{Geilke2011} to the notion of \emph{relational patterns} and \emph{relational pattern languages}. Here, relations other than equality may be used to restrict the valid substitution of patterns. 

This can be achieved by using regular patterns (each variable occurs only once) together with separate relational constraints on variables. The resulting language consists of all words obtained by substitutions satisfying these constraints, thereby extending classical pattern languages to arbitrary relations between variables. Consider, for example, the pattern $\beta = x_1\tc\tc x_2$. Assume $x_1$ and $x_2$ should be substituted by the reversal of each other, i.e., $(h(x_1),h(x_2))\in r_{rev}$, where $h$ is a substitution and $r_{rev}$ is the reversal relation defined by $r_{rev} := \{\ (w,w^R) \mid w\in\Sigma^*\ \}$ (where $w^R$ is the reversal of a word $w$). Then, a substitution $h$ with $h(x_1) = \ta\tb$ and $h(x_2) = \tb\ta$ is called \emph{valid} regarding the pattern $\beta$, we get $h(\beta) = \ta\tb\tc\tc\tb\ta$, and we say that $h(\beta)$ is in the relational pattern language of $\beta$ under the relational constraints that relate $x_1$ and $x_2$ by $r_{rev}$. If we considered a substitution $h'$, however, in which $h'(x_1) = \ta\tb$ and also $h'(x_2)=\ta\tb$, then clearly $h'(x_1) \neq h'(x_2)^R$ and therefore $h'$ is not a valid substitution in this case, resulting in $h'(\beta)$ not being in that relational pattern language of $\beta$.

Geilke and Zilles~\cite{Geilke2011} established general properties of relational pattern languages. For equal length, Holte et al.~\cite{pmlr-v167-holte22a} proved decidability of NE-equivalence for alphabets of size at least three, and Mousawi and Zilles~\cite{Mousawi2024,mousawi2026positivecharacteristicsetsrelational} showed the corresponding E-result. An earlier extension to the binary case is currently claimed to be flawed in their arXiv paper~\cite{mousawi2026positivecharacteristicsetsrelational}. However, they still claim a classification for a large subclass of patterns in the binary setting, while the general case remains open. For reversal, positive characteristic sets exist in the NE case, but no effective construction is known, while no such family exists in the binary E case~\cite{Mousawi2024,mousawi2026positivecharacteristicsetsrelational}. In addition to that, in the cases of two special kinds of relations/constraints, namely regular constraints \cite{Nowotka2024} or length constraints\cite{Nowotka2025}, it was recently shown, among other results, that the still open equivalence problem for erasing pattern languages actually becomes undecidable if these constraints are allowed in addition to equality between variables. This sets a first upper bound of constraints necessary to obtain undecidability for that problem. In comparison to other relations, however, these constraints appear to be very strong, leaving the question of what happens in the cases of weaker relations that do not restrict the substitution space with explicit values. Except the previously stated results, this area seems to be largely unexplored.

This paper considers the following relations: \emph{equality} ($r_{=}$), \emph{equal length} ($r_{|w|}$), \emph{subsequence} ($r_{ssq}$), \emph{abelian equivalence} ($r_{ab}$), \emph{alphabet permutation} ($r_{perm}$), \emph{reversal} ($r_{rev}$), (nonempty) \emph{commutation} ($r_{com^+}$, $r_{com^*}$), the \emph{star relation} ($r_{*}$), and \emph{common (non)emptiness} ($r_{\varepsilon_=}$). These relations are not meant to form an exhaustive list. Rather, they constitute a selection of natural and representative binary relations on words, covering several standard notions of similarity. Many other relations could be studied within this framework (e.g., see discussion in Section~\ref{sec:conclusion}). We leave such extensions for future work. 
We study the equivalence, inclusion, and membership problems under these relations, providing a comprehensive overview of many cases (Table~\ref{tab:complexity-full}).
Section 2 introduces the required notation and models. Sections 3 to 5 address the equivalence, inclusion, and membership problems, respectively. Section 6 discusses open problems and directions for future research. Due to space constraints, some proofs and extended proof sketches are omitted from this proceedings version and are available in the appendix of the full version~\cite{thisPaper} (the corresponding results are marked with (\Rightscissors) and a reference to the respective subsection).

\begin{table}[ht]
    \centering
    \begingroup
    % Compact result/source pointers used only inside this table.
    \newcommand{\resultref}[1]{
        \textsuperscript{\normalfont\scriptsize #1}
    }

    \setlength\tabcolsep{4pt}
    \renewcommand{\arraystretch}{1.2}
    \resizebox{\textwidth}{!}{
    \begin{tabular}{ | C | C | C | C | C | C | C | C | C | C | C | }
        \hline
         & $r_{=}$ 
         & $r_{|w|}$
         & $r_{ssq}$ 
         & $r_{ab}$ 
         & $r_{perm}$ 
         & $r_{rev}$ 
         & $r_{com^*}$
         & $r_{com^+}$
         & $r_{*}$
         & $r_{\varepsilon_=}$\\
        \hline

        E ($\in$)
        & NPC\resultref{a}
        & \textcolor{red}{NPC}\resultref{T15}
        & \textcolor{red}{NPC}\resultref{T15}
        & \textcolor{red}{NPC}\resultref{T15}
        & \textcolor{red}{NPC}\resultref{T15}
        & \textcolor{red}{NPC}\resultref{T15}
        & \textcolor{red}{NPC}\resultref{T15}
        & \textcolor{red}{NPC}\resultref{T15}
        & \textcolor{red}{NPC}\resultref{T15} 
        & \textcolor{red}{NPC}\resultref{T15} \\
        \hline

        E ($\subseteq$)
        & UD\resultref{b}
        & \textcolor{lightgray}{\textbf{?}}
        & \textcolor{red}{UD}\resultref{C8}
        & \textcolor{red}{UD}\resultref{T12}
        & \textcolor{lightgray}{\textbf{?}}
        & \textcolor{red}{UD}\resultref{T12}
        & \textcolor{lightgray}{\textbf{?}}
        & \textcolor{lightgray}{\textbf{?}}
        & \textcolor{red}{UD}\resultref{C8}
        & \textcolor{red}{D}\resultref{T10}\\
        \hline

        E ($=$)
        & \textcolor{lightgray}{\textbf{?}}
        & $\mathrm{D}_{|\Sigma|\geq 3}$\resultref{d}
        & \textcolor{red}{$\geq r_=$}\resultref{P7}
        & \textcolor{lightgray}{\textbf{?}}
        & \textcolor{lightgray}{\textbf{?}}
        & \textcolor{red}{$\geq r_=$}\resultref{P7}
        & \textcolor{lightgray}{\textbf{?}}
        & \textcolor{lightgray}{\textbf{?}}
        & \textcolor{red}{$\geq r_=$}\resultref{P7}
        & \textcolor{red}{D}\resultref{T5}\\
        \hline

        NE ($\in$)
        & NPC\resultref{e}
        & \textcolor{red}{NPC}\resultref{T15}
        & \textcolor{red}{NPC}\resultref{T15}
        & \textcolor{red}{NPC}\resultref{T15}
        & \textcolor{red}{NPC}\resultref{T15}
        & \textcolor{red}{NPC}\resultref{T15}
        & \textcolor{red}{NPC}\resultref{T15}
        & \textcolor{red}{NPC}\resultref{T15}
        & \textcolor{red}{NPC}\resultref{T15} 
        & \textcolor{red}{P}\resultref{C16}\\
        \hline

        NE ($\subseteq$)
        & UD\resultref{b}
        & \textcolor{lightgray}{\textbf{?}}
        & \textcolor{red}{UD}\resultref{C8}
        & \textcolor{red}{UD}\resultref{T9}
        & \textcolor{lightgray}{\textbf{?}}
        & \textcolor{red}{UD}\resultref{T9}
        & \textcolor{lightgray}{\textbf{?}}
        & \textcolor{lightgray}{\textbf{?}}
        & \textcolor{red}{UD}\resultref{C8}
        & \textcolor{red}{P}\resultref{S4.1}\\
        \hline

        NE ($=$)
        & P\resultref{e}
        & $\mathrm{D}_{|\Sigma|\geq 3}$\resultref{c}
        & \textcolor{red}{$?^{\dagger}$}\resultref{L1}
        & \textcolor{red}{P}\resultref{C4}
        & \textcolor{lightgray}{\textbf{?}}
        & \textcolor{red}{$?^{\dagger}$}\resultref{L1}
        & \multicolumn{2}{c|}{
            \textcolor{red}{P}\resultref{C4}
        }
        & \textcolor{red}{$?^{\dagger}$}\resultref{L1}
        & \textcolor{red}{P}\resultref{S3.1}\\
        \hline
    \end{tabular}
    }

    \caption{
        Current state of decision problems for relational pattern languages
        under the considered relations. Red entries are new results.
        NPC = NP-complete, P = polynomial-time decidable,
        D/UD = decidable/undecidable, E/NE = erasing/non-erasing,
        $\geq r_=$ denotes a polynomial-time many-one reduction from
        the equality case, $\dagger$ marks non-trivial necessary
        conditions for equivalence to hold. 
    }
    \label{tab:complexity-full}

    \vspace{2mm}

    \begin{minipage}{.98\textwidth}
        \scriptsize
        \raggedright
        \textit{Pointers.}
        Superscripts T, C, P, L, and S refer to theorem, corollary,
        proposition, lemma, and section numbers in this paper, respectively.\\
        \resultref{a} Jiang et al.~\cite{Jiang1994};
        \resultref{b} Jiang et al.~\cite{Jiang1995},
        Freydenberger and Reidenbach~\cite{Freydenberger2010},
        and Bremer and Freydenberger~\cite{BREMER201215};\\
        \resultref{c} Holte et al.~
        \cite[Theorem~7 and Corollary~16]{pmlr-v167-holte22a};
        \resultref{d} Mousawi and Zilles~
        \cite[Theorem~5]{mousawi2026positivecharacteristicsetsrelational};
        \resultref{e} Angluin~\cite{Angluin1980}.
    \end{minipage}
    \endgroup
\end{table}

\vspace{-1.2cm} % used to remove the big gap between Section 1 and Section 2 after Table 1. Necessary, to keep Section 2 close enough for Section 2 to stop after Page 4 such that Section 3 can start on its own page.

\vspace{0.5cm}

% === PRELIMS === %
\section{Preliminaries}
Let $\N$ denote the natural numbers $\{1, 2, 3, \dots\}$ and let $\N_0 := \N \cup \{0\}$.
For $n,m \in \mathbb{N}$ set $[m,n] := \{k \in \mathbb{N} \mid m \leq k \leq n\}$. 
Denote $[n] := [1,n]$ and $[n]_0 := [0,n]$. 
An \emph{alphabet} $\Sigma$ is a non-empty finite set whose elements are called \emph{letters}.  
A \emph{word} is a finite sequence of letters from $\Sigma$. 
Let $\Sigma^*$ be the set of all finite words over $\Sigma$ and set $\Sigma^+ := \Sigma^* \setminus \{\varepsilon\}$.
We call the number of letters in a word $w \in \Sigma^*$ \emph{length} of $w$, denoted by $|w|$.
Therefore, we have $|\varepsilon| = 0$. 
Denote $\Sigma^k := \{w \in \Sigma^* \mid |w| = k\}$.
For $w \in \Sigma^*$, let $w[i]$ denote $w$'s $i^{th}$ letter for all $i \in [\vert w \vert]$. 
For compactness reasons, we denote $w[i] \cdots w[j]$ by $w[i \cdots j]$ for all $i,j \in [\vert w \vert]$ with $i < j$.
Denote by $|w|_\ta = |\{\ i\in [|w|] \mid w[i] = \ta\}|$ the number of times the letter $\ta$ occurs in $w$.
Set $\al(w) :=  \{\ta \in \Sigma \mid \exists i \in [\vert w \vert] : w[i] = \ta\}$ as $w$’s alphabet.
A \emph{subsequence} $u$ of a word $w$ is a word obtained by deleting arbitrary letters of $w$. Denote by $\SubSeq(w)$ the set of subsequences of $w$. Two words $u,v\in\Sigma^*$ are abelian equivalent, denoted by $u\equiv_{ab}v$, iff $|u|_\ta = |v|_\ta$ for all $\ta\in\Sigma$.

Let $X := \{x_1, x_2, ...\}$ be a countable set of variables such that $\Sigma \cap X = \emptyset$.
A \emph{pattern} is then a non-empty, finite word over $\Sigma \cup X$.
The set of all patterns over $\Sigma \cup X$ is denoted by $Pat_\Sigma$. 
For example, $x_1 \ta x_2 \tb \ta x_2 x_3$ is a pattern over $\Sigma = \{\ta,\tb\}$ with $x_1,x_2,x_3\in X$.
For a pattern $\alpha\in Pat_\Sigma$, let $\var(\alpha) := \{\ x \in X\ |\ |\alpha|_x \geq 1\ \}$ denote the set of variables occurring in $\alpha$.
A \emph{substitution of $\alpha$} is a morphism $h : (\Sigma \cup X)^* \to \Sigma^*$ such that $h(\ta) = \ta$ for all $\ta \in \Sigma$ and 
$h(x) \in \Sigma^*$ for all $x \in X$. 
If we have $h(x) \neq \varepsilon$ for all $x \in \var(\alpha)$, we call $h$ a \emph{non-erasing substitution} for $\alpha$. 
Otherwise $h$ is an \emph{erasing substitution} for $\alpha$. The set of all substitutions w.r.t.~$\Sigma$ is denoted by $H_\Sigma$.
If $\Sigma$ is clear from the context, we may write just $H$.
For a pattern $\alpha\in Pat_\Sigma$, its erasing pattern language $L_E(\alpha)$ and non-erasing pattern language $L_{NE}(\alpha)$
are defined, respectively, by
\begin{align*}
	L_{E}(\alpha) &= \{\ h(\alpha)\ |\ h\in H, h(x) \in \Sigma^* \text{ for all } x\in\var(\alpha)\}, \text{ and } \\
	L_{NE}(\alpha) &= \{\ h(\alpha)\ |\ h\in H, h(x) \in \Sigma^+ \text{ for all } x\in\var(\alpha)\}.
\end{align*}
Let $RegPat_\Sigma = \{ \alpha\in Pat_\Sigma \mid \forall x\in X: |\alpha|_{x} \leq 1\}$ denote the set of all regular patterns, i.e., the set of all patterns where each variable occurs only once. Let $R$ be some set of relations and denote for each $n\in\N$ by $R_n\subset R$ the subset of $n$-ary relations in $R$. Define $Rel_{R} = \{ (r,x_{i_1},... ,x_{i_n}) \mid r\in R_n, x_{i_1},... ,x_{i_n} \in X, i_1,...i_n\in\N, n\in\N \}$ to be the set of all combinations to relate variables in $X$ over the set of relations $R$. Then, a pair $(\alpha,r_\alpha)\in RegPat_\Sigma \times 2^{Rel_R}$ is a relational pattern over the terminal alphabet $\Sigma$ and the set of relations $R$. We denote the set of all relational patterns over $\Sigma$ and $R$ by $\relpat$. For some $(\alpha,r_\alpha)\in \relpat$ and $h\in H$, we say that $h$ is a $r_\alpha$-valid substitution if for all $(r',x_{i_1},...,x_{i_n})\in r_\alpha$ we have $(h(x_{i_1}),...,h(x_{i_n})) \in r'$, i.e., all variables are substituted by words related to each other regarding the relation $r'$. 
The set of all $r_\alpha$-valid substitutions w.r.t.~$\Sigma$ is denoted by $H_{\Sigma,r_\alpha}$.
If $\Sigma$ or $r_\alpha$ are apparent, we may also just write $H_{r_\alpha}$ or $H$. 
For any $(\alpha,r_\alpha)\in\relpat$, we denote its \emph{erasing relational pattern language} $L_E(\alpha,r_\alpha)$ and \emph{non-erasing relational pattern language} $L_{NE}(\alpha,r_\alpha)$, respectively, by
\begin{align*}
	L_E(\alpha,r_\alpha) &= \{\ h(\alpha)\ |\ h\in H_{\Sigma,r_\alpha}\ \} \\
	L_{NE}(\alpha,r_\alpha) &= \{\ h(\alpha)\ |\ h\in H_{\Sigma,r_\alpha} \text{ and } h(x)\in\Sigma^+\text{ for all } x\in\var(\alpha)\ \}.
\end{align*}
In this paper, we primarily investigate the considered decision problems under binary relations, i.e., $|R| = 1$ and, if $R = \{r\}$, then $r$ is 2-ary/binary. In particular, we are interested in the selection of specific customary binary relations given in Table~\ref{tab:relations}. Denote by $\mathcal{R}$ the set of all these relations.

\begin{table}[h!]
    \centering
    \small
    \setlength{\tabcolsep}{6pt}
    \renewcommand{\arraystretch}{1.08}
    \begin{tabularx}{\textwidth}{@{}>{\raggedleft\arraybackslash}p{2.2cm}X@{}}
        \toprule
        \textbf{Relation} & \textbf{Description} \\
        \midrule
        $r_{=}$ 
            & equality of words, i.e., $(u,v)\in r_{=}$ iff $u = v$. \\

        $r_{|w|}$ 
            & length equality of words, i.e., $(u,v)\in r_{|w|}$ iff $|u| = |v|$. \\

        $r_{ssq}$ 
            & subsequence relation, i.e., $(u,v)\in r_{ssq}$ iff $u\in\SubSeq(v)$. \\

        $r_{ab}$ 
            & abelian equivalence of words, i.e., $(u,v)\in r_{ab}$ iff $u \equiv_{ab} v$. \\

        $r_{perm}$ 
            & alphabet permutation of words, i.e., $(u,v)\in r_{perm}$ iff $u$ can be obtained from $v$ by a bijective renaming of the letters. \\

        $r_{rev}$ 
            & reversal of words, i.e., $(u,v)\in r_{rev}$ iff $u = v^R$. \\

        $r_{com^*}$ 
            & commutation between words, i.e., $(u,v)\in r_{com^*}$ iff there exists $z\in\Sigma^*$ such that $u,v\in\{z\}^*$. \\

        $r_{com^+}$ 
            & nonempty commutation between words, i.e., $(u,v)\in r_{com^+}$ iff there exists $z\in\Sigma^+$ such that $u,v\in\{z\}^+$. \\

        $r_{*}$ 
            & word in the Kleene star of the other word, i.e., $(u,v)\in r_{*}$ iff $u\in\{v\}^*$. \\

        $r_{\varepsilon_=}$ 
            & common emptiness or common nonemptiness, i.e., $(u,v)\in r_{\varepsilon_=}$ iff $u = \varepsilon = v$ or $u,v\in\Sigma^+$. \\
        \bottomrule
    \end{tabularx}
    \caption{The set of relations $\mathcal{R}$ considered in this paper.}
    \label{tab:relations}
\end{table}

% === EQUIVALENCE === %
\section{The Equivalence Problem}
We first study equivalence, which in the erasing case already contains the long-standing open problem for $r_=$. In the non-erasing case, we characterize relations for which equivalence is determined by equality of the patterns and of the reflexive, symmetric, and transitive closures of the constraints, covering $r_{ab}$, $r_{com^+}$, and $r_=$. For $r_{rev}$, $r_{ssq}$, and $r_*$, we show that this characterization is insufficient.

In the erasing case, we identify $r_{\varepsilon_=}$ as another relation, besides $r_{|w|}$, with decidable equivalence, while for $r_{rev}$, $r_*$, and $r_{ssq}$, equivalence under $r_=$ polynomial-time many-one reduces to the corresponding equivalence problems.

\subsection{The Non-Erasing Case}

Given any relational pattern $(\alpha,r_\alpha)\in\relpat$, we denote by $[r_\alpha]$ the reflexive, symmetric, and transitive closure on the set of binary related variables in $r_\alpha$. For example, if $(r_=,x_1,x_2),(r_=,x_2,x_3)\in r_{\alpha}$, then we would have $(r_=,x_i,x_j)\in[r_\alpha]$, for any combination $i,j\in\{1,2,3\}$. For the remainder of this section, we assume w.l.o.g. a normal form on the patterns such that all variables are always introduced inside a pattern with increasing indices, starting from 1 (i.e. $x_1$, $x_2$, $x_3$, and so on). The first result gives us a strong condition for relations $r$ that are reflexive on all letters of the alphabet and do not relate two distinct letters.

\begin{lemma}\label{lemma:equivalence-length1-rel-property}
    Assume $|\Sigma|\geq 2$ and let $R = \{r\}$, for a relation $r$ with the following properties: For all $\ta\in\Sigma$, we have $(\ta,\ta)\in r$, and if $(\ta,\tb)\in r$, for two letters $\ta,\tb\in\Sigma$, then $\ta = \tb$. Let $(\alpha,r_\alpha),(\beta,r_\beta)\in\relpat$ be two patterns in normal form. If $L_{NE}(\alpha,r_\alpha) = L_{NE}(\beta,r_\beta)$ then $\alpha = \beta$ and $[r_\alpha] = [r_\beta]$.
\end{lemma}

\begin{proof}
    First, we see that $|\alpha| = |\beta|$ must hold. Suppose otherwise and assume w.l.o.g. that $|\alpha| < |\beta|$. Take some $r_\alpha$-valid substitution $h$ that substitutes each variable by just a single letter. Then $|h(\alpha)| = |\alpha|$. In the NE-case, there can't exist any $r_\beta$-valid substitution $h'$ with $|h'(\beta)| = |\alpha| < |\beta|$. A contradiction.

    Now suppose $\alpha \neq \beta$. As both are regular patterns in normal form, that means there exists a first position $i\in[|\alpha|]$ such that w.l.o.g. $\alpha[i] = \ta$, for some letter $\ta\in\Sigma$, and $\beta[i] = x_j$, for some variable $x_j\in X$. Now let $h$ be a $r_\beta$-valid substitution that substitutes each variable by just a single letter and, in particular, set $h(x_j) = \tb$, for some letter $\tb$ with $\ta\neq\tb$. As $|\alpha| = |\beta|$, any $r_\alpha$-valid substitution $h'$ with $h'(\alpha) = h(\beta)$ would have to match $h'(\alpha[k]) = h(\beta[k])$, for all $k\in[|\alpha|]$, and, in particular, would need to match $h'(\alpha[i]) = h'(\ta) = \ta \neq \tb = h(x_j) = h(\beta[i])$. A contradiction. That concludes $\alpha = \beta$.

    Now suppose $[r_\alpha] \neq [r_\beta]$. Then there exist two positions $i,j\in[|\alpha|]$ such that $\alpha[i] = \beta[i] = x_{i'}$, for some $i'\in[|\alpha|]$, and $\alpha[j] = \beta[j] = x_{j'}$, for some $j'\in[|\alpha|]$, and w.l.o.g. $(r,x_{i'},x_{j'})\in [r_\alpha]$ but $(r,x_{i'},x_{j'})\notin [r_\beta]$. Now choose a substitution $h$ that substitutes each variable just by a single letter such that $h(x_{i'}) = \ta$ and $h(x_{j'}) = \tb$, for two distinct letters $\ta,\tb\in\Sigma$. As $r$ forces all blocks of symmetrically and transitively connected variables to be substituted by the same letter, in the case of length 1 substitutions, we know that $h$ cannot be $r_\alpha$-valid. However, there exists such an $h$ that is $r_\beta$-valid by choosing substitutions of the other variables accordingly. Hence, there exists such an $h(\beta)\in L_{NE}(\beta,r_\beta)$, but we always have $h(\beta)\notin L_{NE}(\alpha,r_\alpha)$. A contradiction.
\end{proof}

This result already confirms a very specific structure for all relations with the above stated properties. In particular, this covers the relations $r_{=}$, $r_{ab}$, $r_{ssq}$, $r_{rev}$, $r_{com^+}/r_{com^*}$ (equivalent for non-erasing pattern languages), and $r_*$. This property, however, is not exhaustive to characterize equivalence. So, adding to the previous property, if we now consider equivalence relations in general, we obtain the following.

\begin{lemma}\label{lemma:equivalence-equivalencerel-langproperty}
    Let $R = \{r\}$ for an equivalence relation $r$. Let $(\alpha,r_\alpha)\in\relpat$ be some relational pattern in normal form. We have $L_{NE}(\alpha,r_\alpha) = L_{NE}(\alpha,[r_\alpha])$.
\end{lemma}
\begin{proof}
    ($\Rightarrow$): First, let $w\in L_{NE}(\alpha,r_\alpha)$. Then there exists some $r_\alpha$-valid substitution $h$ with $h(\alpha) = w$. Let $(r,x_i,x_j)\in r_\alpha$, for two variables $x_i,x_j\in\var(\alpha)$. Then $(h(x_i),h(x_j))\in r$ by definition. As $r$ is an equivalence relation, we also have $(h(x_j),h(x_i))\in r$, $(h(x_i),h(x_i))\in r$, and, if there exists some $(r,x_j,x_k)\in r_\alpha$ (hence $(h(x_j),h(x_k))\in r$), we also have $(h(x_i),h(x_k))\in r$. Hence, $h$ satisfies all constraints in $[r_\alpha]$ and thereby is $[r_\alpha]$-valid, resulting in $w\in L_{NE}(\alpha,[r_\alpha])$.

    ($\Leftarrow$): For the other direction, let $w\in L_{NE}(\alpha,[r_\alpha])$ and let $h$ be an $[r_\alpha]$-valid substitution such that $h(\alpha) = w$. Notice that $r_\alpha\subseteq[r_\alpha]$. Hence, all constraints in $r_\alpha$ are satisfied by $h$, resulting in $h$ being $r_\alpha$-valid. Thus, $w\in L_{NE}(\alpha,r_\alpha)$.
\end{proof}

Combining both previous results, we identify a set of relations for which equivalence of non-erasing languages can be characterized by equality of patterns and the corresponding closures.

\begin{theorem}\label{theorem:equivalence-subst1prop-equivrel-equality}
    Let $R = \{r\}$ for $r$ being an equivalence relation with the following properties: For all $\ta\in\Sigma$, we have $(\ta,\ta)\in r$, and if $(\ta,\tb)\in r$, for two letters $\ta,\tb\in\Sigma$, then $\ta = \tb$. Let $(\alpha,r_\alpha),(\beta,r_\beta)\in\relpat$ be two patterns in normal form. $L_{NE}(\alpha,r_\alpha) = L_{NE}(\beta,r_\beta)$ if and only if $\alpha = \beta$ and $[r_\alpha] = [r_\beta]$.
\end{theorem}
\begin{proof}
    We obtain the first direction immediately from \Cref{lemma:equivalence-length1-rel-property}. For the other direction, assume $\alpha = \beta$ and $[r_\alpha] = [r_\beta]$. By \Cref{lemma:equivalence-equivalencerel-langproperty} we know $L_{NE}(\alpha,r_\alpha) = L_{NE}(\alpha,[r_\alpha])$ and $L_{NE}(\beta,r_\beta) = L_{NE}(\beta,[r_\beta])$. As $\alpha = \beta$ and $[r_\alpha] = [r_\beta]$, we have $L_{NE}(\alpha,[r_\alpha]) = L_{NE}(\beta,[r_\beta])$. Thus, $L_{NE}(\alpha,r_\alpha) = L_{NE}(\beta,r_\beta)$.
\end{proof}

For the relations specifically considered in this paper we thus obtain the following.

\begin{corollary}[{\Rightscissors~\cite[A.1]{thisPaper}}]\label{corollary:equivalence-ne-polydec-somerel}
    Let $R = \{r\}$ for some $r\in\{r_=, r_{ab}, r_{com^+}, r_{com^*}\}$. Given the relational patterns $(\alpha,r_\alpha),(\beta,r_\beta)\in\relpat$, we can check in polynomial time whether $L_{NE}(\alpha,r_\alpha) = L_{NE}(\beta,r_\beta)$.
\end{corollary}

\begin{remark}\label{remark:equivalence-nonerasing-commonempty-linear}
    Notice that the relation $r_{\varepsilon_=}$ has no effect in the non-erasing case. In particular, for any relational pattern $(\alpha,r_\alpha)\in\relpat$ with $R = \{r_{\varepsilon_=}\}$ we have $L_{NE}(\alpha,r_\alpha) = L_{NE}(\alpha)$ for the regular pattern $\alpha$. Hence, equivalence between two such patterns is decidable in linear time in this case~\cite{Angluin1980}.
\end{remark}

The other relations considered in \Cref{lemma:equivalence-length1-rel-property} that are not covered in \Cref{corollary:equivalence-ne-polydec-somerel}, namely $r_{ssq}$, $r_{rev}$, $r_{*}$, are clearly not equivalence relations. Hence, decidability is not immediately clear in these cases. Thus, they probably require a deeper investigation to answer decidability in general. We conjecture they are all decidable.

\subsection{The Erasing Case}
In this part, our primary focus lies on the relation $r_{\varepsilon_=}$. 
First, note that $r_{\varepsilon_=}$ is reflexive, symmetric, and transitive. We extend the notation of the prior section by defining equivalence classes $\mathcal{C}_\alpha = \{C_{\alpha,1},\dots, C_{\alpha,k_\alpha}\}$ of variables $C_{\alpha,i} \subseteq \var(\alpha)$, $i \in [k_\alpha]$ that are in relation with one another. 

In the following, we prove that the equivalence problem for erasing relational pattern languages is indeed decidable in this case. First, we present a general decidability procedure. After that, we identify multiple cases that even result in linear-time and polynomial-time algorithms to decide equivalence.

\begin{theorem}[{\Rightscissors~\cite[A.2]{thisPaper}}]\label{theorem:equivalence-commonEmptiness-decidable-pspace}
    Let $R = \{r_{\varepsilon_=}\}$ and $\Sigma$ be some alphabet. Given the relational patterns $(\alpha,r_\alpha),(\beta,r_\beta)\in\relpat$, checking $L_{E}(\alpha,r_\alpha) = L_{E}(\beta,r_\beta)$ is decidable.
\end{theorem}
\begin{proof}[Proof Sketch.]
First, consider $(\alpha,r_\alpha)$ and for each subset $S \subseteq \mathcal{C}_\alpha$, replace the variables $v \in \bigcup_{C_{\alpha,i} \in S} C_{\alpha,i}$ by any word in $\Sigma^+$ and any other variable $v \notin \bigcup_{C_{\alpha,i} \in S} C_{\alpha,i}$ by $\varepsilon$. We call the resulting language $L_{E}(\alpha,r_\alpha,S)$. Using that, we obtain $L_{E}(\alpha,r_\alpha) = \bigcup_{S \subseteq \mathcal{C}_\alpha} L_{E}(\alpha,r_\alpha,S)$. 
Do the same for $\beta$. Note there exist at most $2^{k_\alpha} \leq 2^{|\alpha|}$ subsets $S$ for $\alpha$ and at most $2^{k_\beta}\leq 2^{|\beta|}$ subsets for $\beta$. For each $S \subseteq \mathcal{C}_\alpha$, we can take a sub-pattern $\alpha_S$ of $\alpha$ that only contains the variables that are not substituted by the empty word. As $\alpha_S$ is a regular pattern, it can be observed that $L_E(\alpha,r_\alpha,S) = L_{NE}(\alpha_S)$ and we can construct an NFA $A_{\alpha_S}$ with $L(A_{\alpha_S}) = L_{NE}(\alpha_S) = L_{E}(\alpha,r_\alpha,S)$. Doing that for each such $S$, combining all these NFAs into one big NFA $A_\alpha$ for $L_E(\alpha,r_\alpha)$ and doing the same for $(\beta,r_\beta)$ in an NFA $A_{\beta}$, we can compare the languages $L_{E}(\alpha,r_\alpha)$ and $L_{E}(\beta,r_\beta)$ by comparing $L(A_\alpha) = L(A_\beta)$, concluding the result.
\end{proof}

Next, we introduce approaches with an improved running time for some pattern structures.

\begin{lemma}[{\Rightscissors~\cite[A.3]{thisPaper}}]\label{lemma:empty-together-poly-equivalence}
    Let $R = \{r_{\varepsilon_=}\}$ and $\Sigma$ be some alphabet. Let $(\alpha,r_\alpha),(\beta,r_\beta)\in\relpat$ be some relational patterns.
    \begin{enumerate}
        \item If $\alpha$ and $\beta$ are \emph{terminal-free}, we can check in linear time whether $L_{E}(\alpha,r_\alpha) = L_{E}(\beta,r_\beta)$.
        \item If $|\Sigma| = 1$, we can check in linear time whether $L_{E}(\alpha,r_\alpha) = L_{E}(\beta,r_\beta)$.
        \item If there exists $c \in \Sigma$ that is not a terminal in $\alpha$ and $\beta$, we can check in polynomial time whether $L_{E}(\alpha,r_\alpha) = L_{E}(\beta,r_\beta)$.
    \end{enumerate}
\end{lemma}
\begin{proof}[Proof Sketch.]
The key observation is that every non-empty equivalence class contributes at least one symbol, while otherwise variables are unrestricted. Hence, in 1. the accepted languages of $\alpha$ and $\beta$ are completely determined by the minimum size of an equivalence class (and, in 2., additionally by the number of terminals). These quantities can be computed by a single scan over the patterns.
In 3., the patterns are decomposed into terminal and variable blocks, and each equivalence class is represented by its occurrence vector across the variable blocks. By removing dominated classes, we obtain a minimal pattern in polynomial time. The two languages are equivalent if and only if these minimal patterns coincide.
\end{proof}

Finally, we consider some other relations in the context of the equivalence problem for erasing relational pattern languages. We can show that this problem must be at least as hard as in the case of $r_=$ for the relations $r_{*}$, $r_{ssq}$, and $r_{rev}$, in terms of time complexity. While the first two follow trivially by simulating $r_=$, the last one requires a more technical construction.

\begin{proposition}\label{proposition:equivalence-equalityhardrelations}
    Let $R = \{r\}$, $r\in\{r_{ssq},r_*,r_{rev}\}$ and $(\alpha,r_\alpha),(\beta,r_\beta)\in\relpat$. The  equivalence problem for erasing pattern languages polynomial-time many-one reduces to the equivalence problem for erasing relational pattern languages over $R$.
\end{proposition}
\begin{proof}
We reduce from the case $R={r_=}$. Let $(\alpha,r_\alpha)$ be a relational pattern over equality. For $r\in\{r_*,r_{ssq}\}$, replace each constraint $(r_=,x_i,x_j)\in r_\alpha$ by $(r,x_i,x_j)$ and $(r,x_j,x_i)$. This preserves valid substitutions, since equality is equivalent both to being mutual subsequences and to being mutual powers. Applying this change to both input patterns, respectively, gives the reductions for $r_*$ and $r_{ssq}$.

For $r=r_{rev}$, introduce for every $x_i\in\var(\alpha)$ a fresh variable $x_i'$. Let $\alpha^{\leftarrow}$ be the reversal of $\alpha$ with each $x_i$ replaced by $x_i'$, and set $\alpha^{\leftrightarrow}:=\alpha\alpha^{\leftarrow}$. Define $r_{\alpha^{\leftrightarrow}}$ by adding $(r_{rev},x_i,x_i')$ for all $x_i\in\var(\alpha)$ and $(r_{rev},x_i,x_j'),(r_{rev},x_j,x_i')$ for every $(r_=,x_i,x_j)\in r_\alpha$. Then
\[ L_E(\alpha^{\leftrightarrow},r_{\alpha^{\leftrightarrow}}) = \{ww^R\mid w\in L_E(\alpha,r_\alpha)\}. \]
Indeed, the constraints $(r_{rev},x_i,x_i')$ force $x_i'$ to be substituted by the reversal of $x_i$, while each constraint $(r_{rev},x_i,x_j')$ enforces $h(x_i)=h(x_j)$ using $h(x_j')=h(x_j)^R$. Conversely, every $r_\alpha$-valid substitution can be extended by setting $h(x_i'):=h(x_i)^R$. Thus, the constructed language is exactly the image of $L_E(\alpha,r_\alpha)$ under the injective map $f(w)=ww^R$. Applying the construction to both input patterns therefore preserves language equality. The construction is polynomial-time computable.
\end{proof}

This concludes our examination of the equivalence problem. For the non-erasing cases, we either obtain efficient decidability results or structural restrictions that may support future decidability proofs. For the erasing cases, however, most questions remain open, indicating that even for relatively weak structure-enforcing relations, equivalence is far from trivial.

% === INCLUSION === %
\section{The Inclusion Problem}
In this section, we continue with the inclusion problem for relational pattern languages. We recall that the inclusion problem for non-erasing as well as erasing pattern languages over $r_=$ is generally undecidable for all alphabets $\Sigma$ of size $|\Sigma| \geq 2$, as shown and refined in ~\cite{Jiang1995,Freydenberger2010,BREMER201215, Bremer2025Melting}. As we will see in this section, for several of the considered relations in $\mathcal{R}$, we also obtain that this problem is undecidable in the non-erasing or erasing cases. In addition to that, we also identify decidable cases.

First, consider the relations $r_{ssq}$ and $r_*$. As we observed before, we can simulate the exact behavior of $r_=$ (e.g., see notes in the beginning of the proof of Proposition~\ref{proposition:equivalence-equalityhardrelations}). 
By the fact that the inclusion problem is undecidable for $r_=$ in both the erasing and non-erasing cases, we obtain the following.

\begin{corollary}\label{thm:eqrels}\cite{Jiang1995,Freydenberger2010,BREMER201215}
    Let $\Sigma$ be any finite alphabet with $|\Sigma|\geq 2$. Let $r\in\{r_=,r_{ssq},r_*\}$ and let $R = \{r\}$. Given $(\alpha,r_\alpha),(\beta,r_\beta)\in\relpat$, it is undecidable to answer whether $L_\mathcal{X}(\alpha,r_\alpha)\subseteq L_\mathcal{X}(\beta,r_\beta)$, for $\mathcal{X}\in\{E,NE\}$.
\end{corollary}

Now, we continue with the analysis for relations for which we cannot immediately infer the undecidability from such trivial relationships.

\subsection{The Non-Erasing Cases}
In the non-erasing case, first notice that, similar to the equivalence problem, in the case of $r_{\varepsilon_=}$, we compare two regular patterns (as $\varepsilon$ substitutions are excluded anyways). Here, inclusion is known to be decidable in polynomial time~\cite{Shinohara1982,Terada2005}. For the relations $r_{ab}$ and $r_{rev}$ we observe something interesting. Reconstructing the reduction from \cite{BREMER201215} in the setting of relational pattern languages, the inclusion problem 
for NE relational pattern languages under $r_{ab}$ or $r_{rev}$ is indeed also undecidable, even though $r_{ab}$ generally implies a significantly less restrictive behavior than $r_=$, and $r_{rev}$ implies a generally distinct structure.

\begin{theorem}[{\Rightscissors~\cite[B.1]{thisPaper}}]\label{prop:inclusion_non-erasing_abel-undec}
    Let $R := \{r\}$ for $r\in\{r_{ab},r_{rev}\}$ and $|\Sigma| \geq 2$ for some alphabet $\Sigma$. For relational patterns $(\alpha,r_\alpha),(\beta,r_\beta)\in\relpat$, it is undecidable whether $L_{NE}(\alpha,r_\alpha) \subseteq L_{NE}(\beta,r_{\beta})$.
\end{theorem}
\begin{proof}[Proof Sketch.]
    We adapt the construction of Bremer and Freydenberger~\cite{BREMER201215} for the undecidability of inclusion for non-erasing pattern languages over a binary alphabet. The key observation is that, on the words used in their construction, the relations $r_{ab}$ and $r_{rev}$ can simulate equality whenever needed. If two related variables are forced to be substituted by single letters, then both relations force these letters to be equal. And, if one of two related variables is substituted by a unary word, then the other one must be substituted by the same unary word as well. Thus, all equality tests used in the original construction can be reproduced by $r_{ab}$ and $r_{rev}$ on the relevant parts of the pattern.

    We first consider a binary alphabet $\{0,\#\}$ and write $\bar{r}\in\{r_{ab},r_{rev}\}$. Let $U$ be the universal Turing machines used in~\cite{BREMER201215}. We construct relational patterns $(\alpha,r_\alpha)$ and $(\beta,r_{\beta})$ such that
    \[ L_{NE}(\alpha,r_\alpha) \subseteq L_{NE}(\beta,r_\beta) \]
    if and only if $U$ has no accepting computation. Hence, the undecidable emptiness problem for universal Turing machines reduces to the inclusion problem for non-erasing relational pattern languages under $\bar{r}$.

    The pattern $\beta$ is the relational-pattern analogue of the pattern from ~\cite{BREMER201215}. It has the form
    \[ \beta := a_1b_1\ \#^5\ a_2x_{1,1}x_{2,1}\cdots x_{\mu,1}b_2\ \#^5\ r_1\hat{\beta}_1r_2\hat{\beta}_2\cdots r_\mu\hat{\beta}_\mu r_{\mu+1} \]
    where, $a_j,b_j,r_k,x_{\ell,1}$ are variables, for $j\in[2]$, $k\in[\mu+1]$, $\ell\in[\mu]$, and for each $i\in[\mu]$,
    \[ \hat{\beta}_i := 0x_{i,2}\cdots x_{i,5}0\ \upgamma_i\ 0x_{i,6}\cdots x_{i,9}0\ \updelta_i\ 0x_{i,10}\cdots x_{i,13}0. \]
    for variables $x_{i,j}$, $j\in[2,13]$, and patterns $\upgamma_i,\updelta_i$. For now, the relevant constraints in $r_{\beta}$ are
    \[ (\bar r,a_1,a_2),\quad (\bar r,b_1,b_2),\quad (\bar r,x_{i,j},x_{i,j'}) \quad\text{for all } i\in[\mu]\text{ and }j,j'\in[13]. \]
    The variables $r_1,\ldots,r_{\mu+1}$ are unconstrained. The patterns $\upgamma_i$ and $\updelta_i$ are the components from ~\cite{BREMER201215} encoding the possible reasons why a string is not an accepting computation encoding. In particular, $\upgamma_i$ and $\updelta_i$ contain no factor $\#^4$. Moreover, for each $i\in[\mu]$, all variables occurring in $\upgamma_i\updelta_i$ are not related to any other variable occurring in some other  $\upgamma_j\updelta_j$ for $i\neq j$ and $i,j\in[\mu]$, nor to any variable $x_{k,\ell}$, for $k\in[\mu]$ and $\ell\in[13]$, as well as the variables $a_j$ and $b_j$.

    The pattern $\alpha$ is defined by
    \[ \alpha := 0^{\mu+1}\ \#^5\ 0^\mu\#0^\mu\ \#^5\ tv0\alpha_10v0\alpha_20vt \]
    where $v = 0\#^4 0$, and $\alpha_1$, $\alpha_2$ are relational pattern analogues to the ones used in ~\cite{BREMER201215}. In particular, $\alpha_1$ contains no factor $\#^3$, while $\alpha_2$ contains no occurrence of $\#$ at all. We have $r_\alpha = \emptyset$. Finally, $t$ is obtained from the suffix $r_1\hat{\beta}_1\cdots r_\mu\hat{\beta}_\mu r_{\mu+1}$ by replacing every variable by $0$.

    It remains to explain why this construction still works with $\bar{r}$ instead of equality. The factors $\#^5$ force any matching of a word generated by $\alpha$ against $\beta$ to match each of the three by $\#^5$ divided parts against each other. In particular, the middle block $0^\mu\# 0^\mu$ has to match
    \[ a_2x_{1,1}x_{2,1}\cdots x_{\mu,1}b_2. \]
    Since $a_2$ and $b_2$ are related to variables that match unary words in the first block, they are forced to behave as in the equality construction. Consequently, exactly one variable $x_{i,1}$ is substituted by $\#$, while all other variables $x_{k,1}$ are substituted by $0$. Because all variables $x_{i,j}$, $j\in[13]$, are mutually related, the selected component $\hat{\beta}_i$ receives the marker $\#$ in all $x_{i,j}$ positions, while all other components $\hat{\beta}_k$, $k\neq i$, receive a $0$ in all $x_{k,j}$ positions. Thus, the construction selects exactly one component $\hat{\beta}_i$ to match against $v0\alpha_10v0\alpha_20v$ in $\alpha$.

    This selected component checks one possible error regarding encodings of accepting computations of $U$. More precisely, the matching forces 
    \[ h(0\alpha_10)=h'(\upgamma_i) \quad\text{and}\quad h(0\alpha_20)=h'(\updelta_i) \]
    for some $r_\beta$-valid substitution $h'$. Hence, a word generated by $\alpha$ is covered by $\beta$ exactly when the encoded computation violates at least one of the conditions of accepting computations, encoded in some component $\hat{\beta}_i$. Hence, an accepting computation encoding satisfies none of these conditions, therefore resulting in a word in $L_{NE}(\alpha,r_\alpha)$ that is not in $L_{NE}(\beta,r_\beta)$.

    The next difference from the equality-based construction is the use of related variables inside the components $\upgamma_i$ and $\updelta_i$. These variables are used either to propagate selected single letters such as $0$ and $\#$, or to compare unary words that encode counter values. Both cases are handled by $r_{ab}$ and $r_{rev}$ analogously to $r_=$. Thus, each component $\hat{\beta}_i$ can be constructed analogously as in ~\cite{BREMER201215}. Therefore
    \[ L_{NE}(\alpha,r_\alpha) \subseteq L_{NE}(\beta,r_\beta) \]
    holds if and only if $U$ has no accepting computation. This proves undecidability over the binary alphabet. For larger alphabets, the extension from ~\cite{BREMER201215} also applies. Additional components $\hat{\beta}_i$ can be introduced to handle the occurrence of any other letter than $0$ or $\#$, making a selection of the first $\mu$ components only necessary if these letters do not occur in a substitution of $\alpha$. The full formal verification of the adapted construction, due to its extensive length, is given in the full version of this paper~\cite{thisPaper}.
\end{proof}

\subsection{The Erasing Cases}
\label{subsec:inc-era}
We continue with erasing relational pattern languages. Similar to the non-erasing cases, we show that the inclusion problem is decidable under $r_{\varepsilon_=}$ and undecidable under the reversal relation $r_{rev}$ as well as under abelian equivalence $r_{ab}$. Indeed, in the case of $r_{\varepsilon_=}$, we can use the exact same approach as for the equivalence problem from Theorem~\ref{theorem:equivalence-commonEmptiness-decidable-pspace} to obtain decidability using the construction of two NFAs that represent $L_{E}(\alpha,r_\alpha)$ and $L_{E}(\beta,r_\beta)$ respectively.

\begin{theorem}
    Let $R = \{r_{\varepsilon_=}\}$. Given the relational patterns $(\alpha,r_\alpha),(\beta,r_\beta)\in\relpat$, checking $L_{E}(\alpha,r_\alpha) \subseteq L_{E}(\beta,r_\beta)$ is decidable.
\end{theorem}

Again, for some pattern structures, we introduce some approaches with improved polynomial running time.

\begin{lemma}[{\Rightscissors~\cite[B.2]{thisPaper}}]\label{lemma:empty-together-poly-inclusion}
    Let $R = \{r_{\varepsilon_=}\}$. Given some alphabet $\Sigma$ and relational patterns $(\alpha,r_\alpha),(\beta,r_\beta)\in\relpat$.
    \begin{enumerate}
        \item If $\alpha$ and $\beta$ are \emph{terminal-free}, we can check in linear time whether $L_{E}(\alpha,r_\alpha) \subseteq L_{E}(\beta,r_\beta)$.
        \item If $|\Sigma| = 1$, we can check in linear time whether $L_{E}(\alpha,r_\alpha) \subseteq L_{E}(\beta,r_\beta)$.
    \end{enumerate}
\end{lemma}
\begin{proof}[Proof Sketch.]
    By the same arguments as in \Cref{lemma:empty-together-poly-equivalence}, inclusion in the terminal-free case holds if and only if the minimum class size of $\alpha$ is at least that of $\beta$.
    For unary alphabets, the language is additionally determined by the number of terminals. Hence, inclusion holds if and only if $\alpha$ has at least as many terminals as $\beta$ and its minimum equivalence class size is at least that of $\beta$. In both cases, we can check the properties in linear time.
\end{proof}

For the undecidability proofs, we adapt the framework of Freydenberger and Reidenbach~\cite{Freydenberger2010}, which shows undecidability in the erasing case under $r_=$ for all alphabets $\Sigma$ with $|\Sigma|\geq 2$. The required adaptations for $r_{ab}$ and $r_{rev}$ rely on the same principle as in Theorem~\ref{prop:inclusion_non-erasing_abel-undec} (related single letters and unary words enforce equality). We therefore give only the main ideas for \Cref{thm:uirev} below. The full proof, due to its extensive length, is provided in the full version~\cite{thisPaper} for verification reasons.

\begin{theorem}[{\Rightscissors~\cite[B.3]{thisPaper}}]\label{thm:uirev}
    Let $\bar{r} \in \{r_{rev},r_{ab}\}$ and let $R = \{\bar{r}\}$. For all alphabets $\Sigma$ with $|\Sigma|\geq 2$, given $(\alpha,r_{\alpha}),(\beta,r_{\beta}) \in \relpat$, it is generally undecidable whether $L_E(\alpha,r_\alpha) \subseteq L_E(\beta,r_{\beta})$.
\end{theorem} 
\begin{proof}[Proof Sketch.]
    Let $A$ be a nondeterministic 2-counter automaton (see ~\cite{thisPaper} for a definition - for this sketch we can work with $A$ abstractly and just acknowledge that the emptiness problem for nondeterministic 2-counter automata is undecidable). Consider, first, the following construction under the equality relation. Let
    \[\alpha_A := v\ v\ \#^6\ v\ x\ v\ y\ v\ \#^6\ v\ u\ v\]
    where $x,y$ are distinct variables, $v = 0\#^30$ and $u = 0\#^50$ and 
    \[\beta_A:=x_1x_{1'}\ \ldots x_{\mu}x_{\mu'}\#^6\hat{\beta}_1\ldots\hat{\beta}_{\mu}\#^6\ddot{\beta_1}\ldots\ddot{\beta_{\mu}}\]
    with, for all $i \in \{1, \ldots, \mu\}$, and $\ddot{\beta_i} := x_{i_4}\ \eta_i\ x_{i_5}$, where $x_i$, $x_{i'}$, $x_{i_4}$, $x_{i_5}$, $\ldots$ are distinct variables with $(r_=,x_i, x_{i'})$, $(r_=,x_i, x_{i_1})$, $(r_=,x_i, x_{i_2})$, $(r_=,x_i,x_{i_3})$, $(r_=,x_i,x_{i_4})$, $(r_=,x_i,x_{i_5}) \in r_{\beta_A}$ and all $\eta_i \in X^*$ are terminal-free patterns.
    We have
    $\eta_i := z_{i}\ \hat{z}_{i_1}\ \hat{z}_{i_2}\ \hat{z}_{i_3}\ \hat{z}_{i_4}\ \hat{z}_{i_5}\ z_{i'}$ with distinct variables $z_i,z_{i'},\hat{z}_{i_1},\hat{z}_{i_2},\hat{z}_{i_3},\hat{z}_{i_4},\hat{z}_{i_5}$ and $(r_=,z_i,z_{i'})$, $(r_=,\hat{z}_{i_1},\hat{z}_{i_2})$, $(r_=,\hat{z}_{i_1},\hat{z}_{i_3})$, $(r_=,\hat{z}_{i_1},\hat{z}_{i_4})$, $(r_=,\hat{z}_{i_1},\hat{z}_{i_5}) \in r_{\beta_A}$.
    
    The main concept of the proof of \Cref{thm:uirev} is to show that $L_E(\alpha_A,r_{\alpha_A}) \subseteq L_E(\beta_A,r_ {\beta_A})$ iff $A$ has no accepting computation. All possible properties of non-accepting computations are listed in the $\hat{\beta}_1 \ldots \hat{\beta}_{\mu}$ part where each $\hat{\beta}_i, i \in \{1, \ldots, \mu \}$ encodes property of a non-accepting computation together with selecting variables $x_{i_j}$ for $j \in {1,2,3}$. The part of $\beta_A$ and $\alpha_A$ before the first $\#^6$ ensures that exactly one pair $x_i,x_i', i \in \{1, \ldots, \mu\}$ is selected, resulting in that exactly one non-accepting computation is selected and that exactly one $\ddot{\beta}_i, i \in \{1, \ldots, \mu\}$ is selected. With this, we get that $\eta_i, i \in \{1, \ldots, \mu\}$ matches $u$ and thus, that $z_i = 0$ and $\hat{z}_{i_1} = \#$. We use these two variables to encode the $0$ and the $\#$ in the (non-accepting) computation encodings. If we can show that this selecting process also works in the setting $R = \{\bar{r}\}$ for $\bar{r} \in \{r_{rev},r_{ab}\}$, i.e., we can construct two variables with which we can encode the letters $0$ and $\#$, and we can model an exact match for two unary words to match counter value encodings of configuration of $A$ inside some $\hat{\beta}_i$, then this proof also works in this setting. For two variables $x,y$, we have if $x = 0$ (or $x = \#$) and $(x,y) \in \bar{r}$ for $\bar{r} \in \{r_{rev},r_{ab}\}$, then $y = 0$ (or $y = \#$ respectively). Thus, we can model with the relation $r_{rev}$ as well as with the relation $r_{ab}$ identical variables if we know a one letter substitution of one variable. Moreover, we get for two variables $x,y$ with $(x,y) \in \bar{r}$ for $\bar{r} \in \{r_{rev},r_{ab}\}$ that $x=y$ must hold if $x$ or $y$ gets substituted by a unary word.

    If we want $L_E(\alpha_A,r_{\alpha_A}) \subseteq L_E(\beta_A,r_{\beta_A})$, we need to match the factor $v\ v$ to $x_1x_{1'} \ldots x_{\mu}x_{\mu'}$ and $v\ u\ v$ to $\ddot{\beta}_1 \ldots \ddot{\beta}_{\mu}$ due to our delimiter $\#^6$. We start with the match of $v\ v$ to $x_1x_{1'} \ldots x_{\mu}x_{\mu'}$. We have $v\ v = 0\ \#^3 0\ 0 \#^3 0$. Since the variables are pairwise in relation, we can always assume that we start substituting with the left variable. If a variable gets substituted to something different from $\varepsilon$, then its corresponding variable also gets substituted to something different from $\varepsilon$. For $r_{rev}$, we get the following reasonable possibilities for the beginning of the match:
    \begin{table}[h!]
        \centering
            \begin{tabular}{l|l|l}
                $x_i$ &  $x_{i'}$ & Possible prefix for $v\ v = 0\#^300\#^30$? \\
                \hline
                 $0$ & $0$ & $00 \neq 0\#$ \\
                 $0\#$ & $\# 0$ & $\ 0\#^20 \neq 0\#^3$ \\
                 $0\#^2$ & $\#^20$ & $0\#^4 0 \neq 0\#^300$ \\
                 $0\#^3$ & $\#^30$ &  $\ 0\#^60 \neq 0\#^300\#^2$ \\
                 $0\#^30$ & $0\#^3 0$ & $0 \#^300\#^30 = v\ v$
            \end{tabular}
        \label{tab:posvrev}
    \end{table}
    
    \noindent For $r_{ab}$, we get the following reasonable possibilities for the beginning of the match:
    \begin{table}[h!]
        \centering
            \begin{tabular}{l|l|l}
                $x_i$ &  $x_{i'}$ & Possible prefix for $v\ v = 0\#^300\#^30$? \\
                \hline
                 $0$ & $0$ & $00 \neq 0\#$ \\
                 $0\#$ & $0\#$ or $\# 0$ & $0\#0\#,\ 0\#^20 \neq 0\#^3$ \\
                 $0\#^2$ & $0\#^2$ or $\#0\#$ or $\#\#0$ & $0\#^20\#^2,\ 0\#^30\#,\ 0\#^4 0 \neq 0\#^300$ \\
                 $0\#^3$ & $0\#^3$ or $\#0\#^2$ or $\#^20\#$ or $\#^30$ &  $0\#^30\#^3,\ 0\#^3\#0\#^2,\ 0\#^50\#,\ 0\#^60 \neq 0\#^300\#^2$ \\
                 $0\#^30$ & $0\#^3 0$ & $0 \#^300\#^30 = v\ v$
            \end{tabular}
        \label{tab:posv}
    \end{table}
    
    Thus, the only possibility is that for exactly one $i \in \{1, \ldots, \mu\}$ the variables $x_i$ and $x_{i'}$ get substituted with $v$ and all other variables $x_j,x_{j'}$ for $j \neq i, j \in \{1, \ldots, \mu\}$ get substituted with $\varepsilon$. This together with the matching from $v\ u\ v$ to $\ddot{\beta}_1 \ldots \ddot{\beta}_{\mu}$ implies that $\eta_i$ needs to match $u$. For this matching between $\eta_i = z_{i}\ \hat{z}_{i_1}\ \hat{z}_{i_2}\ \hat{z}_{i_3}\ \hat{z}_{i_4}\ \hat{z}_{i_5}\ z_{i'}$ and $u = 0\#^50$, we get the following possibilities for both relations, $r_{ab}$ and $r_{rev}$ (recall that $(\bar{r},z_i,z_{i'}),(\bar{r},\hat{z}_{i_1},\hat{z}_{i_2}),(\bar{r},\hat{z}_{i_1},\hat{z}_{i_3}),(\bar{r},\hat{z}_{i_1},\hat{z}_{i_4}),(\bar{r},\hat{z}_{i_1},\hat{z}_{i_5}) \in r_{\beta_A}$ in this setting):
    
    \begin{table}[ht]
        \centering
        \begin{tabular}{l|l|l}
           $z_i$  & $\hat{z}_{i_1}$ & Possible substitution for $u = 0\#^50$? \\ \hline
            $\varepsilon$ & $0$ or $\#$ & $0^5,\ \#^5 \neq 0\#^50$ \\
            $0$ & $\#$ & $0\#^50 = 0\#^5 0 = u$ \\
            $0\#$ & $\#$ & $0\#\#^5\#0 \neq 0 \#^50$ \\
            $0\#\#$ & $\#$ & $0\#^2\#^5\#^20 \neq 0\#^50$ \\
            $0\#\#\#$ & $\varepsilon$ & $0\#^60 \neq 0\#^50$
        \end{tabular}
        \label{tab:posu}
    \end{table}

    Thus, the only possibility is that $z_i$ and $z_{i'}$ get substituted with $0$ and $\hat{z}_{i_k}, k \in \{1,2,3,4,5\}$ get substituted with $\#$. With this, we can encode the $0$ and the $\#$ and thus, the proof of Freydenberger and Reidenbach~\cite{Freydenberger2010} works in this setting analogously.
\end{proof}

This concludes our results for the inclusion problem. For the remaining cases $r_{|w|}$, $r_{perm}$, $r_{com^*}$, and $r_{com^+}$, decidability remains open, as the reduction properties used above seem not to carry over: $r_{|w|}$ and $r_{perm}$ do not enforce equality of letters or unary words over at least binary alphabets, while $r_{com^*}$ or $r_{com^+}$ do not even enforce equal lengths. Thus, the decidability boundary for structurally weaker relations remains open.

% === MEMBERSHIP === %
\section{The Membership Problem}
In this final section, we briefly discuss the membership problem -- also known as the matching problem -- for relational pattern languages for all relations in $\mathcal{R}$ and general patterns.
For the case of $r_{=}$, NP-completeness of the membership problem was originally shown in \cite{Angluin1980} and \cite{Jiang1994} for the non-erasing and erasing cases, respectively. In fact, independent communities have shown this under several considered restrictions (cf. \cite{FERNAU2015287,Manea2019}). As shown in this section, all but one case remain NP-complete.

\begin{proposition}\label{prop:membership-innp-relational}
   Let $r\in\mathcal{R}$ and $R = \{r\}$. Let $w\in\Sigma^*$ and $(\alpha,r_\alpha)\in\relpat$. Deciding $w\in L_{X}(\alpha,r_\alpha)$, for $X\in\{E,NE\}$, is in NP.
\end{proposition}

Indeed, NP-containment follows by the known arguments from the cases arising in $r_=$. We can guess a substitution $h$ with $|h(\alpha)| = |w|$, check the relational constraints in polynomial time, and finally check whether $h(\alpha) = w$. Regarding NP-hardness, we revisit the exposition given in \cite{Manea2019}, where a reduction is given from the known NP-complete positive 1-in-3-SAT (no negated variables, 3-CNF clauses, satisfiable with exactly one literal per clause set to true) problem. For example, assuming $\Sigma = \{0,\#\}$, the formula $$\{\{X_1,X_2,X_3\},\{X_2,X_4,X_5\},\{X_1,X_4,X_5\},\{X_3,X_4,X_6\}\}$$ is in 3-CNF and 1-in-3 satisfiable if and only if the word $w = 0\#0\#0\#0$ is in the erasing pattern language of the pattern $\alpha = x_1x_2x_3\#x_2x_4x_5\#x_1x_4x_5\#x_3x_4x_6$. Using $w = 0^4\#0^4\#0^4\#0^4$, this construction works for the NE case as well. For most of the relations considered in this paper, this framework suffices.

\begin{lemma}[{\Rightscissors~\cite[C.1]{thisPaper}}]\label{lemma:membership-trivial-hardness}
    Let $r\in\{r_=, r_{|w|}, r_{ssq}, r_{ab}, r_{perm}, r_{rev}, r_{*},r_{com^+},r_{\varepsilon_=}\}$ s.t. $R = \{r\}$ and let $|\Sigma|\geq 2$. For $w\in\Sigma^*$, $(\alpha,r_{\alpha})\in\relpat$, deciding $w\in L_{E}(\alpha,r_\alpha)$ is NP-hard. If $r \notin \{r_{com^+},r_{\varepsilon_=}\}$, deciding $w\in L_{NE}(\alpha,r_\alpha)$ is NP-hard.
\end{lemma}

This leaves us with the cases of $r_{com^+} (r_{com^*})$ in the non-erasing case and with $r_{com^*}$ in the erasing case. For these, a slightly different reduction is necessary but possible to achieve. As the general idea is still a natural reduction from satisfiability, in this case the variant 3-CNF-SAT (or just 3-SAT), we omit the formal proof to ~\cite{thisPaper}. This concludes the following.

\begin{theorem}[{\Rightscissors~\cite[C.2]{thisPaper}}]\label{theorem:membership-np-complete-cases}
    Let $r\in\mathcal{R}$ s.t. $R = \{r\}$ and assume $|\Sigma|\geq 2$. For $w\in\Sigma^*$ and $(\alpha,r_{\alpha})\in\relpat$, deciding $w\in L_{E}(\alpha,r_\alpha)$ is NP-complete. If $r \neq r_{\varepsilon_=}$, deciding $w\in L_{NE}(\alpha,r_\alpha)$ is NP-complete.
\end{theorem}

Finally, we recall from \Cref{remark:equivalence-nonerasing-commonempty-linear} that the relation $r_{\varepsilon_=}$ has no effect in the non-erasing case, resulting in each pattern language $L_{NE}(\alpha,r_\alpha)$, for $(\alpha,r_\alpha)\in\relpat$ and $R = \{r_{\varepsilon_=}\}$, being equal to the pattern language $L_{NE}(\alpha)$ of the regular pattern $\alpha$. The membership problem for regular patterns has been shown to be decidable in linear time by Shinohara in \cite{Shinohara1983}. Hence, the following immediately holds.

\begin{corollary}\label{corollary:membership-poly-case}
    Let $R = \{r_{\varepsilon_=}\}$. For $w\in\Sigma^*$ and $(\alpha,r_{\alpha})\in\relpat$, we can decide in linear time whether $w\in L_{NE}(\alpha,r_\alpha)$.
\end{corollary}

This concludes our discussion of the membership problem. For alphabets of size at least two, except for $r_{\varepsilon_=}$ in the non-erasing case, all considered cases are NP-complete. Thus, hardness persists even for structurally weak relations such as $r_{\varepsilon_=}$, $r_{com^*}$, and $r_{com^+}$, confirming that matching words against compact relational pattern descriptions is generally computationally hard.

% === CONCLUSION === %
\section{Final Discussion}
\label{sec:conclusion}
We considered the three main decision problems equivalence, inclusion, and membership problems for relational pattern languages under various typical binary relations, i.e., $r_=$, $r_{|w|}$, $r_{ssq}$, $r_{ab}$, $r_{perm}$, $r_{rev}$, $r_{com^+}$, $r_{com^*}$, $r_{*}$, and $r_{\varepsilon_=}$. For all problems, we obtained a selection of combinatorial and complexity results that can serve as a basis for several settings where encoding problems into the relational pattern languages with different relations might be a natural choice. In addition to that, these results together with the remaining open cases discussed in the respective sections can serve as a basis for future research in all the directions considered in this paper.

In addition to these results, the problems considered in this paper could also be considered for other natural relations that have not been considered in this context, yet. Among these, for example, the non-equal relation that allows for everything but equal substitutions among related variables is an interesting candidate that also finds relevance in the literature (e.g., see \cite{Manea2020}). There, it has been shown that it is NP-complete to decide whether a word $w$ can be factorized into $k$ distinct factors. By constructing a pattern $x_1...x_k$ and relating all variables with each other using non-equality, this directly reduces to the relational pattern language setting using that relation, resulting in NP-completeness of Membership. It would be interesting to know how the inclusion and equivalence problems behave with respect to this relation. Another candidate could be a cyclic shift relation that restricts substitutions of related variables to conjugates of a word. As an equivalence relation, we expect similar results as obtained in this paper to also be obtainable in this case. 

Finally, we only considered patterns using a single relation. Allowing multiple relations could increase expressiveness substantially: decidable single-relation cases might become undecidable when relations are combined, and open cases such as erasing equivalence under equality might become undecidable as well when additional relations weaker than regular or length constraints are allowed.

\subsubsection*{Acknowledgments.}
We thank the anonymous reviewers for their valuable comments and suggestions.

\bibliographystyle{eptcsini}
\bibliography{20-bibliography}

\end{document}